\documentclass[letterpaper, 10 pt, conference]{ieeeconf}  

\IEEEoverridecommandlockouts                              
\def\R{\mathbb{R}}
\def\C{\mathbb{C}}

\usepackage{soul}
\usepackage{amsthm}
\usepackage{amssymb}  
\usepackage{amsfonts}
\usepackage{mathrsfs}
\usepackage{mathtools}
\usepackage[bbgreekl]{mathbbol}
\usepackage[noadjust]{cite}

\usepackage{booktabs} 
\usepackage{colortbl}
\usepackage{multirow} 
\usepackage{multicol}

\DeclareSymbolFontAlphabet{\mathbb}{AMSb}
\DeclareSymbolFontAlphabet{\mathbbl}{bbold}

\DeclareMathOperator{\diag}{diag}
\DeclareMathOperator{\blkdiag}{blkdiag}
\DeclareMathOperator{\trace}{trace}
\DeclareMathOperator{\diff}{d}

\DeclareMathOperator{\vect}{vec}
\DeclareMathOperator{\vech}{vech}
\DeclareMathOperator{\vechrs}{vechrs}

\DeclareMathOperator*{\argmin}{arg\,min}

\DeclareMathOperator{\rank}{rank}

\newcommand{\ju}{\mathrm{j}}

\usepackage{siunitx}
\AtBeginDocument{\DeclareSIUnit{\MWh}{MWh}}
\AtBeginDocument{\DeclareSIUnit{\kWh}{kWh}}
\usepackage{amsmath}
\usepackage{mathtools}
\usepackage{isomath}

\newtheorem{proposition}{Proposition}

\newtheorem{assumption}{Assumption}

\newtheorem{definition}{Definition}
\newtheorem{theorem}{Theorem}

\newtheorem{lemma}{Lemma}
\newtheorem{problem}{Problem}

\newtheorem{example}{Example}
\usepackage{lipsum}

\newcommand{\slow}[1] 
{ \textcolor{red}{(SL: #1)}}  
\newcommand{\ognjen}[1]  
{ \textcolor{ForestGreen}{(OS: #1)}}  
\newcommand{\luc}[1]  
{ \textcolor{blue}{(LW: #1)}}  

\title{\LARGE \bf
Tractable Identification of Electric Distribution Networks
}

\author{Ognjen Stanojev, Lucien Werner, Steven Low and Gabriela Hug
\thanks{This research was supported by the Swiss National Science Foundation under NCCR Automation, grant agreement 51NF40\_180545.}
\thanks{Ognjen Stanojev and Gabriela Hug are with the Power Systems Laboratory, ETH Z\"{u}rich, Z\"{u}rich, Switzerland, emails: \{ognjens, ghug\}@ethz.ch.}%
\thanks{Lucien Werner and Steven Low are with the Department of Computing and Mathematical Sciences at California Institute of Techonology, Pasadena, California, USA, emails: \{lwerner, slow\}@caltech.edu.}
        }%

\begin{document}

\maketitle
\thispagestyle{empty}
\pagestyle{empty}

\begin{abstract}
The identification of distribution network topology and parameters is a critical problem that lays the foundation for improving network efficiency, enhancing reliability, and increasing its capacity to host distributed energy resources. Network identification problems often involve estimating a large number of parameters based on highly correlated measurements, resulting in an ill-conditioned and computationally demanding estimation process. We address these challenges by proposing two admittance matrix estimation methods. In the first method, we use the eigendecomposition of the admittance matrix to generalize the notion of stationarity to electrical signals and demonstrate how the stationarity property can be used to facilitate a maximum a posteriori estimation procedure. We relax the stationarity assumption in the second proposed method by employing Linear Minimum Mean Square Error (LMMSE) estimation. Since LMMSE estimation is often ill-conditioned, we introduce an approximate well-conditioned solution. Our quantitative results demonstrate the improvement in computational efficiency compared to the state-of-the-art methods while preserving the estimation accuracy.
\end{abstract}

\section{Introduction}
Electric distribution networks are a vital component of the energy infrastructure, serving as the final layer in power delivery to residential and commercial users. The increasing integration of renewable energy sources and the implementation of decarbonization policies require modernization of the present control and monitoring practices in power distribution systems. The admittance matrix is at the heart of numerous power system analysis techniques, including (optimal) power flow, state estimation, and short circuit analysis \cite{Admittance2019}. It bears the structure of graph Laplacian matrices \cite{AGTENetworks}, thus unambiguously explaining the network topology and related line parameters. However, distribution utilities often lack accurate topology and parameter information, hindering the construction of the admittance matrix and the use of the available analysis tools \cite{Tutorial2022}.

The recent installation of a significant number of micro Phasor Measurement Units ($\mu$PMUs) \cite{vonMeier2017} and smart meters \cite{Schirmer2023} in distribution grids provides network operators with high-precision and high-sampling-rate measurements. These data streams enhance the observability of distribution grids and enable network identification.
In general, network identification problems involve determining network connectivity (i.e., topology) \cite{Tutorial2022,Bolognani2013VStats,Deka2018StructureNetworks,Cavraro2019Probing}, line parameters \cite{Wehenkel2020}, or both \cite{MoffatUnsupervised2020,IPF2022,Fabbiani2022IdentificationExperiment,Ardakanian2019,Yu2018PaToPa,Gupta2021_CompoundMeasurements,Brouillon2021BayesianNetworks,JS_TSG2022}, using bus voltage, current injection, or branch flow measurements. In this work, we address the problem of estimating the admittance matrix using $\mu$PMU measurements of bus voltage and current injection phasors. Although prior research has examined various network identification problems, we only provide a focused overview of the literature on estimating the admittance matrix in the following paragraph. Please refer to \cite{Tutorial2022} for a more detailed review.

In \cite{MoffatUnsupervised2020}, matrix least squares estimation is applied on phasor measurements to approximate the admittance matrix. A constrained least squares approach is developed in \cite{IPF2022} to enforce the Laplacian matrix structure in the least squares estimate. Instead of batch processing, a recursive least squares method is developed in \cite{Fabbiani2022IdentificationExperiment} to enable frequent online updates. In \cite{Ardakanian2019}, a sparsity promoting $\ell_1$-norm regularizer is introduced to enhance the least squares estimation when the admittance matrix is known to be sparse. An alternative approach to promote sparsity is proposed in \cite{Yu2018PaToPa}, where lines with small conductance values are progressively removed after performing the least squares estimation. The works above assume noise-free measurements of the independent variables, which leads to biased estimates when using realistic data with errors in all measurements (variables).
This limitation of the least squares approaches can be successfully tackled by error-in-variables methods, such as total least squares \cite{Gupta2021_CompoundMeasurements}. A weighted total least squares method is introduced in \cite{Brouillon2021BayesianNetworks} and then extended in \cite{JS_TSG2022} to a Bayesian framework that allows exploiting different forms of prior knowledge of the admittance matrix, thus creating a flexible framework that can achieve high estimation accuracy. 

Despite previous methods laying a solid foundation for admittance matrix estimation, challenges involving ill-conditioning and high computational and memory requirements in the estimation process remain unaddressed. Poor conditioning is common to least squares approaches \cite{MoffatUnsupervised2020,IPF2022,Ardakanian2019} and arises even in the estimation of small-size networks due to high correlations in voltage or current measurements. 
Significant computational burden and memory requirements arise when solving the weighted total least squares \cite{Brouillon2021BayesianNetworks, JS_TSG2022} since a substantial number of measurements are required for accuracy, and a large number of parameters contained within the admittance matrix need to be estimated.

In this paper, we address the aforementioned challenges by proposing two admittance matrix estimation methods that are computationally efficient and numerically stable. The first method is motivated by the recent developments in the graph signal processing community \cite{DSPonGraphs2013} on the identification of graph filters \cite{StationarySPG2017,StationaryGP2017}. Expanding on these works, we use the eigendecomposition of the admittance matrix to generalize the notion of stationarity to electrical signals in power networks with a constant reactance-resistance ratio. Subsequently, we demonstrate how a Maximum a Posteriori (MAP) estimation method (resembling \cite{Brouillon2021BayesianNetworks,JS_TSG2022}) can be simplified when the current injections are stationary.

In the second proposed method, we relax the adopted assumptions and consider a Linear Minimum Mean Square Error (LMMSE) estimation method which is applicable to general power networks and generic current statistics. The solution to LMMSE is known as the Wiener filter, which may suffer from poor conditioning and might not respect the Laplacian structure of the admittance matrix. To address the ill-conditioning issue, we introduce an approximate solution based on eigenvalue truncation. Furthermore, we demonstrate that the Laplacian structure can be enforced via a postfiltering procedure without significant additional computational effort. Previous works which examined the voltage and current injection statistics in a similar way focused on topology identification rather than admittance matrix estimation \cite{Bolognani2013VStats,Deka2018StructureNetworks}. 

\textbf{Notation}. We denote the sets of real and complex numbers by $\R$ and $\C$. Given a matrix $A$, $A^\top$ denotes its transpose, $A^*$ denotes the entrywise conjugate, and $A^\mathsf{H}=(A^{*})^\top$ denotes the conjugate transpose. For column vectors $x\in\C^n$ and $y\in\C^m$, we use $(x,y)\coloneqq [x^\top,y^\top]^\top\in\C^{n+m}$ to denote a stacked vector. For a random vector $X$, we use $\mathbb{E}[X]$ to denote its mean and $\Sigma_X$ to denote its covariance matrix. Finally, $\mathcal{I}_n$ denotes the $n\times n$ identity matrix, $\mathbbl{1}_n$ and $\mathbbl{0}_n$ are $n$-dimensional vectors of all ones and zeros, respectively.

\section{Admittance Matrix Model of Power Grids} \label{sec:2}
Consider a static (steady-state), single-phase equivalent distribution network composed of $n$ nodes $\mathcal{N}=\{1,\dots,n\}$ and $m$ undirected branches $\mathcal{E}\subseteq\mathcal{N}\times\mathcal{N}$. The network is modeled as a connected and undirected graph $\mathcal{G}\coloneqq(\mathcal{N},\mathcal{E},\mathcal{W})$, with complex-valued edge weights $\mathcal{W}\coloneqq\{y_{ij}\in\C : y_{ij}= g_{ij}+\ju b_{ij},g_{ij}>0, b_{ij}\leq0,\forall\{i,j\}\in\mathcal{E}\}$ representing the series admittances in the standard lumped $\pi$-model of a transmission line. For the purpose of defining the incidence matrix of $\mathcal{G}$, let us assign to each edge a unique identifier $e=\{1,\dots,m\}$ and an arbitrary orientation. Admittances connected to the ground are defined by $y_{i0}\coloneqq g_{i0}+\ju b_{i0}$, with $g_{i0}\geq0, b_{i0}\geq0,\forall i\in\mathcal{N}$, and referred to as shunt admittances. 
Each node $i\in\mathcal{N}$ in the network is associated with a nodal current injection $\bar{I}_i\in\C$ and a nodal voltage $\bar{V}_i\in\C$.  Kirchhoff's and Ohm's laws lead to the following model of the considered electric network~\cite{AGTENetworks}:
\begin{equation} \label{eq:admittance_model}
    \bar{I} = {B}\diag(\{y_e\}_{e=1}^m){B}^\top \bar{V} + \diag(\{y_{i0}\}_{i=1}^n)\bar{V} = {Y} \bar{V},
\end{equation}
where $\bar{V}\coloneqq(\bar{V}_1,\dots,\bar{V}_{n})$ and $\bar{I}\coloneqq(\bar{I}_1,\dots,\bar{I}_{n})$ collect the bus voltages and the current injections, respectively, and $B\in\{-1,0,1\}^{n\times m}$ is the node-edge incidence matrix of $\mathcal{G}$.
\begin{definition}\label{def:1}
    The admittance matrix $Y$ is a complex symmetric matrix, with diagonal elements given by $Y_{ii} = \sum_{j=1,j\neq i}^n y_{ij} + y_{i0},\forall i\in\mathcal{N}$, and off-diagonal elements defined by $Y_{ij}=-y_{ij},\forall\{i,j\}\in\mathcal{E}$ and $Y_{ij}=0$ otherwise.
\end{definition}
Under the adopted assumptions on the real and imaginary parts of the series and shunt admittances, the necessary and sufficient condition for the invertibility of $Y$ is the existence of at least one shunt admittance.
The assumptions made are reasonable for distribution networks, and we refer the reader to \cite{SLowBook,Turizo2022} for a broader discussion on the invertibility of $Y$.

Therefore, singular admittance matrices under 
the adopted assumptions have zero row sums, i.e., $Y\mathbbl{1}_n=\mathbbl{0}_n$ iff $Y$ is singular.
The linear map defined by the singular admittance matrix $Y:\C^n\rightarrow \C^n$ has the nullspace of dimension one consisting of vectors in $\mathrm{span}(\mathbbl{1}_n)=\{\alpha \mathbbl{1}_n, \alpha\in\C\}$. 
Hence, the Moore-Penrose pseudoinverse of $Y$, denoted by $Y^\dagger$, can be used to form a subspace of solutions to \eqref{eq:admittance_model}, given by
\begin{equation} \label{eq:impedance_model}
    \bar{V} = Y^\dagger \bar{I} + \alpha\mathbbl{1}_n,
\end{equation}
where $\alpha\in\C$. The preceding relationship holds if and only if the current injections are balanced, that is, $\mathbbl{1}_n^\top \bar{I} = 0$. The same relationship can be used for invertible $Y$ in which case $\bar{I}\in\C^n$ is unrestricted and $\alpha=0$.
Equation \eqref{eq:impedance_model} represents the so-called \textit{impedance matrix} model of the network.

It is evident from Definition~\ref{def:1} that $Y$ is not necessarily a \textit{normal matrix} since it is non-Hermitian complex symmetric. Hence, it is not always unitarily diagonalizable. Nevertheless, if the conductance-susceptance ratio\footnote{The conductance-susceptance ratio is more commonly referred to as the reactance-resistance ratio or the ``$x/r$ ratio" in the power systems literature. These terms are used interchangeably in this work.} is identical for all lines across the network, then $Y$ is \textit{normal} \cite{SLowBook}. This assumption is typically valid for lines at the same voltage level.

\begin{theorem}[Pseudoinverse of $Y$] \label{th:1}
Suppose the conductance-susceptance ratio $\frac{g_{ij}}{b_{ij}}$
is the same for all $(i,j)\in\mathcal{E}$.  Then
\begin{enumerate}
\item $Y$ is a normal matrix and has a spectral decomposition 
$Y = \mathbbl{W}\mathbbl{\Lambda}\mathbbl{W}^\mathsf{H}$, where $\mathbbl{\Lambda}$ is the diagonal matrix with the eigenvalues of $Y$ 
on its diagonal, 
and the columns of $\mathbbl{W}$ are the corresponding eigenvectors.

\item The Moore-Penrose pseudoinverse of $Y$ is $Y^\dagger = \mathbbl{W}\mathbbl{\Lambda}^\dagger\mathbbl{W}^\mathsf{H}$, where $\mathbbl{\Lambda}^\dagger$ is the diagonal matrix obtained from $\mathbbl{\Lambda}$ by replacing nonzero eigenvalues of $Y$ by their 
reciprocals.
\end{enumerate}
\end{theorem}

\section{Network Identification Problem}
\subsection{Problem Formulation}
We take a statistical perspective and consider zero-mean random vectors $I\coloneqq \bar{I}-\mathbb{E}[\bar{I}]$ and $V\coloneqq \bar{V}-\mathbb{E}[\bar{V}]$, defined such that their means satisfy $\mathbb{E}[I]=Y\mathbb{E}[V]$. Let us further define the voltage covariance matrix $\Sigma_V=\mathbb{E}[{V}{V}^\mathsf{H}]$, the current injection covariance matrix $\Sigma_I=\mathbb{E}[{I}{I}^\mathsf{H}]$, and their cross-covariance  $\Sigma_{{I}{V}} = \mathbb{E}[{I}{V}^\mathsf{H}]$.
The considered network identification problem is formally stated in the following.

\begin{problem}\label{prob:1}
Given a set $\mathcal{S}\coloneqq\{(\tilde{{V}}^1,\tilde{{I}}^1),\dots,(\tilde{{V}}^N,\tilde{{I}}^N)\}$ of $N$ pairs of noisy, zero-centered bus voltage and current injection measurements pertaining to different steady-state operating points, the objective is to infer the underlying distribution network -- its edges and the associated admittance values -- or equivalently, the admittance matrix $Y$ (Def.~\ref{def:1}). 
\end{problem}

The $N$ measurements of voltage and current phasors
are further collected in matrices
    $\tilde{\mathbbl{V}} = \begin{bmatrix} \tilde{V}^1 & \tilde{V}^2 & \dots & \tilde{V}^N \end{bmatrix}$ and  $\tilde{\mathbbl{I}} = \begin{bmatrix} \tilde{{I}}^1 & \tilde{{I}}^2 & \dots & \tilde{{I}}^N \end{bmatrix}$
that will be used to characterize solutions to Problem~\ref{prob:1}. We are putting forth the following two assumptions to limit the scope of our analysis.
\begin{assumption}
    The measurement matrices are assumed to be full rank, i.e., $\rank(\tilde{\mathbbl{V}})=n$ and $\rank(\tilde{\mathbbl{I}})=n$.
\end{assumption}
\begin{assumption}
The nodes $j$ where the current is neither injected nor extracted ($I_j=0$) have been removed \textit{a priori} by applying the Kron \cite{Kron2013} or subKron \cite{MoffatUnsupervised2020} reduction. 
\end{assumption}
The first assumption implies that $N\geq n$. Additionally, experiment design \cite{DU202013311} might be required to guarantee that the measurement matrices are full rank. The second assumption implies that all the nodes $i$ at which $I_i\neq0$ are collected in $\mathcal{N}$ and are assumed to be observed. Note that Kron reduction of the network may lead to a non-sparse admittance matrix, and the Kron-reduced graph $\mathcal{G}$ may not be a tree graph, which are common assumptions in distribution network studies \cite{Cavraro2019Probing}.

\subsection{Measurement Model}
The available $\mu$PMU or smart meter measurements are corrupted by measurement noise. We adopt a generic linear statistical model to represent the individual bus voltage $\tilde{V}\in\C^n$ and current injection $\tilde{I}\in\C^n$ observations as follows:
\begin{equation}\label{eq:factor_analysis_model}
    \tilde{V} = V + \varepsilon_v, \qquad \tilde{I} = I + \varepsilon_i, 
\end{equation}
where $\varepsilon_v$ and $\varepsilon_i$ are complex random vectors describing the measurement noise. 
We assume that the noise vectors follow uncorrelated complex multivariate Gaussian distributions with zero mean: $\varepsilon_v \sim \mathcal{NC}(\mathbbl{0}_n,\sigma_v^2\mathcal{I}_n), \varepsilon_i \sim \mathcal{NC}(\mathbbl{0}_n,\sigma_i^2\mathcal{I}_n)$. The covariances $\sigma_v$ and $\sigma_i$ might be time-varying, but there is no temporal or spatial correlation in the measurement noise. Note that other noise models may be applicable \cite{Brouillon2021BayesianNetworks,Varghese2022TransmissionNoise}. 

\section{Network Identification under Stationary Current Injections}
This section extends the standard notions of wide-sense stationarity in discrete time to define stationarity with respect to $Y$ for networks that satisfy the following assumption. 
\begin{assumption}\label{ass:3}
The conductance-susceptance ratio $\frac{g_{ij}}{b_{ij}}$ is the same
for all lines $(i,j)\in \cal E$ across the network.
\end{assumption}
It is first demonstrated that the eigenvectors of $Y$ can be identified from the voltage covariance matrix when the current injections satisfy the stationarity property. Subsequently, we employ a MAP procedure to estimate the corresponding eigenvalues, hence identifying $Y$ according to Theorem~\ref{th:1}.

\subsection{Recovering the Eigenvectors of $Y$}
\begin{definition}\label{def:stationarity}
Given an admittance matrix $Y$ with a spectral decomposition $Y=\mathbbl{W}\mathbbl{\Lambda}\mathbbl{W}^\mathsf{H}$, a zero-mean random variable $X$ is said to be Wide-Sense Stationary (WSS) with respect to $Y$ if its covariance matrix $\Sigma_X$ also has a spectral decomposition with
the unitary $\mathbbl{W}$, i.e., $\Sigma_X=\mathbbl{W}\mathbbl{\Lambda}_X\mathbbl{W}^\mathsf{H}$ where $\mathbbl{\Lambda}_X$ is a diagonal matrix of eigenvalues of $\Sigma_X$.
\end{definition}
A practically relevant example of a variable that is WSS with respect to $Y$ is white noise $W$, characterized by $\mathbb{E}[W]=0$ and $\mathbb{E}[WW^\mathsf{H}]=\sigma^2\mathcal{I}_n$. Current injections in distribution grids are predominantly determined by loads that can reliably be modeled as white noise over short time intervals (on the order of seconds). Hence, the assumption of white noise current injection statistics has been common in the distribution network identification literature \cite{Bariya2021GuaranteedGrids}. 
The following proposition establishes a connection between the voltage statistics and the admittance matrix when the network is subjected to stationary current injections.
\begin{proposition}
Let $I$ be WSS with respect to $Y$  
and $I=YV$. Then $V$ is also WSS with respect to $Y$.
\end{proposition}
\begin{proof}
Since $V$ is already defined to be zero-mean, we only need to show that 
the covariance matrix $\Sigma_V$ of $V$ is 
unitarily diagonalized by $\mathbbl{W}$, which is derived by considering that
\begin{align*}
    \Sigma_V &= \mathbb{E}[VV^\mathsf{H}]=\mathbb{E}[Y^\dagger I(Y^\dagger I)^\mathsf{H}]=Y^\dagger \Sigma_I (Y^\dagger)^\mathsf{H}\\ &=\mathbbl{W}\mathbbl{\Lambda}^\dagger\mathbbl{W}^\mathsf{H}\mathbbl{W}\mathbbl{\Lambda}_I\mathbbl{W}^\mathsf{H}(\mathbbl{W}\mathbbl{\Lambda}^\dagger\mathbbl{W}^\mathsf{H})^\mathsf{H}\\&=\mathbbl{W}(|\mathbbl{\Lambda}^\dagger|^2\mathbbl{\Lambda}_I)\mathbbl{W}^\mathsf{H},
\end{align*}
where $\Sigma_I = \mathbbl{W}\mathbbl{\Lambda}_I\mathbbl{W}^\mathsf{H}$
is the covariance of $I$.
Hence, $\Sigma_V$ is unitarily diagonalized by $\mathbbl{W}$, which concludes the proof.
\end{proof}

Remarkably, the eigenvectors of the voltage covariance matrix $\Sigma_V$ are the eigenvectors of the admittance matrix $Y$ given WSS current injections.
In general, we cannot verify if $I$ is WSS since the admittance matrix is unknown. However, in a practically relevant case when $I$ is white noise, stationarity holds trivially, and the eigendecomposition of $\Sigma_V$ can be performed to identify the eigenvectors of $Y$. On the other hand, the eigenvalues cannot be recovered similarly since only their magnitude can be computed from the above decomposition, i.e., from $\mathbbl{\Lambda}_V=|\mathbbl{\Lambda}^\dagger|^2\mathbbl{\Lambda}_I$, but not the phase.

\subsection{Maximum a Posteriori Estimation}
Upon recovering the eigenvectors of $Y$, maximum a posteriori estimation can be leveraged to determine the eigenvalues in $\diag{(\mathbbl{\Lambda})}$. The MAP estimate gives the most likely choice of the latent variables $(V,I,Y)$ given the observations $(\tilde{{V}},\tilde{{I}})$. To this end, the posterior distribution can be formulated using Bayes' rule and the conditional independence axioms:
\begin{align} \label{eq:posterior_distribution}
    p(V,I,Y|\tilde{{V}},\tilde{{I}}) &\propto  p(\tilde{{V}}|{V},{Y})p(\tilde{{I}}|{I},{Y})\frac{p({V},{I})}{p(\tilde{{V}},\tilde{{I}})}p({Y}) \nonumber\\
    &\,\mathrm{s.t.} \quad {I} = {YV},
\end{align}
where the admittance matrix is assumed to be independent of the electric variables and their measurements. According to the measurement model in \eqref{eq:factor_analysis_model}, the distributions $p(\tilde{{V}}|{V},{Y})$ and $p(\tilde{{I}}|{I},{Y})$ are Gaussian and can be expressed using the change of variables formula. For simplicity, priors on voltages and currents are considered noninformative, thus represented as uniform distributions over their respective domains. Under this assumption, the quotient of priors ${p({V},{I})}/{p(\tilde{V},\tilde{I})}$ can be neglected. Finally, a prior commonly imposed on $Y$ assumes a unit variance Gaussian distribution on all entries of $Y$. Such prior can be represented by a matrix Gaussian distribution $p(Y) = \exp{({-\trace{(YY^\mathsf{H}})})}$ and leads to ridge regularization. An elaborate discussion on other practically relevant prior distributions $p(Y)$ is given in \cite{JS_TSG2022}. The negative $\log$ minimization of the posterior distribution given in \eqref{eq:posterior_distribution} is constructed, resulting in
\begin{align} \label{eq:original_opt}
    \min_{I,V,Y}&\,\,\,\|\tilde{{V}}-{V} \|_2^2+\|\tilde{{I}}-{I}\|_2^2+\beta \|Y\|_{\mathrm{F}}^2 \\
    \mathrm{s.t.}&\quad  {I}={YV},\nonumber
\end{align}
where $\beta>0$ is a constant regularization parameter proportional to the measurement noise variance. 
The problem at hand is nonconvex 
and is characterized by a large number of decision variables. Furthermore, previous works \cite{JS_TSG2022} apply vectorization of the admittance matrix as a part of the solution approach, which further increases the scale of the problem. We next demonstrate how the formulation can be simplified by leveraging the obtained spectral template $\mathbbl{W}$. 

Changing the coordinates to the orthonormal basis consisting of the columns
of $\mathbbl W$, the bus voltage and current injection vectors are defined as ${\nu}\coloneqq\mathbbl{W}^\mathsf{H}{V}$ and ${\varphi}\coloneqq\mathbbl{W}^\mathsf{H}{I}$.
Then, (i) the network model \eqref{eq:admittance_model} reduces to $\varphi = \mathbbl{\Lambda}\nu$ and the model \eqref{eq:impedance_model} to $\nu = \mathbbl{\Lambda}^\dagger\varphi$;
(ii) given that the Frobenius norm is unitarily invariant, the regularization term can be reformulated as $\|Y\|_\mathrm{F}=\|\mathbbl{W}^\mathsf{H}Y\mathbbl{W}\|_\mathrm{F}=\|{\lambda}\|_2$, where $\lambda=\diag(\mathbbl{\Lambda})$;
(iii) the complex power loss\footnote{An interesting physical interpretation can be given to the voltage and current representations based on this property. By taking the complex magnitude of the loss, we obtain $|S_\mathrm{loss}|= \sum_{i=1}^n|\lambda_i^*||\nu_i|^2$. The components in $\nu$ weighted by higher magnitude eigenvalues contribute more to the power loss magnitude. Thus, the magnitudes of entries in $\nu$ inform the complex power loss-efficiency of a steady-state operating point.} in the network is given by 
\begin{align*}
S_\mathrm{loss} & \ = \ 
\trace\left( VI^\mathsf{H} \right) 
\ = \ \trace\left( \mathbbl{W} \left(\nu \nu^{\sf H}\right) 
\mathbbl{\Lambda}^{\sf H} \mathbbl{W}^{\sf H}  \right)
\\ &
\ = \ \trace \left( \left( \nu \nu^{\sf H}\right) \mathbbl{\Lambda}^{\sf H} \right)
\ = \ \trace \left( \mathbbl{\Lambda}^{\dag} \left( \varphi \varphi^{\sf H} \right) \right),
\end{align*}
that is, $S_\mathrm{loss} = \ \sum_{i=1}^n \lambda_i^* |\nu_i|^2 = \ \sum_{i=1}^n \lambda_i^\dagger |\varphi_i|^2$.

The optimization in \eqref{eq:original_opt} is
equivalent to the following problem in voltages
$\nu$, currents $\varphi$, and eigenvalues $\lambda$ of $Y$:
\begin{subequations} \label{eq:opt_reformulated}
\begin{align}
    \min_{{\varphi},{\nu},{\lambda}}&\,\,\,\|\tilde{{V}}-{\mathbbl{W}{\nu}} \|_2^2+\|\tilde{{I}}-{\mathbbl{W}{\varphi}}\|_2^2+\beta\|\lambda\|_2^2 \\
    \mathrm{s.t.}&\quad  \varphi = \mathbbl{\Lambda}\nu,
\end{align}
\end{subequations}
with a convex objective function and bilinear constraints 
$\mathbbl{\Lambda}\nu$ enforcing the network model. 
The classical algorithm for solving this problem is the alternate block coordinate descent \cite{BlockCD2011}, which alternates between setting $({\varphi},{\nu})$ constant and solving for ${\lambda}$, and vice versa, until convergence. Therefore, the two optimization problems that need to be iteratively solved are given by
\begin{align}
    \hat{\lambda} &= \argmin\, \|\varphi-\diag{(\nu)}\lambda\|_2^2 + \beta\|\lambda\|_2^2, \label{eq:nnls} \\
    \hat{\nu} &= \argmin \|\tilde{{V}}-{\mathbbl{W}{\nu}} \|_2^2+\|\tilde{{I}}-{\mathbbl{W}\mathbbl{\Lambda}{\nu}}\|_2^2, \label{eq:quad_prog}
\end{align}
where \eqref{eq:quad_prog} is a quadratic program and \eqref{eq:nnls} is a regularized least-squares problem, both solvable in closed-form as
\begin{align}
    \hat{\lambda} &= \big(\beta\mathcal{I}_n+\diag{(|\nu|)}\big)^{-1}\diag{(\varphi\nu^\mathsf{H})}, \label{eq:nnls_sln}\\
    \hat{\nu} &= \big(\mathcal{I}_n+\mathbbl{\Lambda}^2\big)^{-1}(\mathbbl{W}^\mathsf{H}\Tilde{V}+\mathbbl{\Lambda}\mathbbl{W}^\mathsf{H}\Tilde{I}). \label{eq:qp_sln}
\end{align}
The matrices subject to inversion are guaranteed to be invertible since both $\lambda^2$ and $|\nu|$ are nonnegative and $\beta>0$.

\section{Wiener Filter-based Network Identification}
In this section, we consider a more general approach when
the assumptions of current injection stationarity (Def.~\ref{def:stationarity}) and constant $x/r$ ratio (Assumption~\ref{ass:3}) used in the previous section may not hold. To this end, let us  define ${Z}\coloneqq({{I}},{{V}})$, with the corresponding joint covariance matrix given by
\begin{equation}
    {\Sigma_{Z}} = \begin{bmatrix} {\Sigma_{{I}}} & {\Sigma_{{I}{V}}} \\
    {\Sigma_{{I}{V}}}^\mathsf{H} & {\Sigma_{{V}}}\end{bmatrix}.
\end{equation} In linear minimum mean square error estimation, the aim is to estimate ${I}$ from ${V}$ using a filter ${Y}$ such that the estimate ${Y}{V}$ minimizes the mean square error $\mathbb{E}[\|{V}{Y}-{I}\|_2^2]$. Assuming that ${\Sigma}_{{{V}}}$ is full rank, the Wiener-Hopf equation admits a simple closed-form solution, namely 
\begin{equation} \label{eq:WF}
    {Y_\mathrm{W}}\coloneqq {\Sigma_{IV}}{\Sigma_{{V}}}^{-1},
\end{equation}
called the Wiener filter. The corresponding minimum mean square error matrix is the Shur complement of $\Sigma_V$ in the joint covariance matrix, that is $\Sigma_I - \Sigma_{IV}\Sigma_V^{-1}\Sigma_{IV}^{\mathsf{H}}$.
The quality of the Wiener filter estimate might be degraded by the effect of additive noise in the current injection and bus voltage measurements~\eqref{eq:factor_analysis_model}. Furthermore, large condition numbers of ${\Sigma_{{V}}}$ are commonly encountered and might hinder the numerical computation of ${\Sigma_{{V}}}^{-1}$. These issues are addressed in the next subsection.

\subsection{A Well-Conditioned Wiener Filter Approximation}
A square matrix is ill-conditioned if it is invertible but becomes singular for a small perturbation of some of its entries. More formally, given a \textit{normal} matrix $A\in\C^{n\times n}$ the condition number $\kappa(A)=|\lambda_\mathrm{max}(A)|/ |\lambda_\mathrm{min}(A)|$ is the ratio of its largest eigenvalue $\lambda_\mathrm{max}(A)$ to its smallest eigenvalue $\lambda_\mathrm{min}(A)$ by moduli. If $\kappa(A)$ is high, $A$ is said to be ill-conditioned.
We begin our analysis by illustrating the physical nature of the conditioning issue in network identification. 

\begin{figure}[!b]
    \centering
    \includegraphics[scale=0.875]{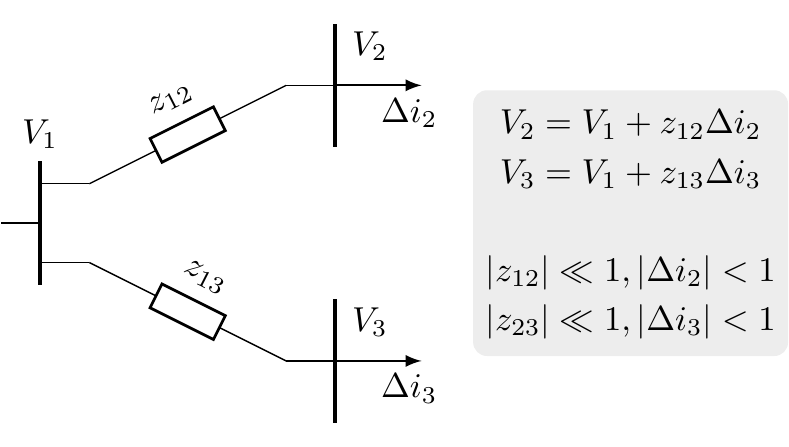}
    \caption{A 3-bus example demonstrating the physical origin of poor conditioning of the covariance matrix $\Sigma_V$.}
    \label{fig:illconditioning_diagram}
\end{figure}
\begin{example}
Let us consider a simple 3-bus example in Fig.~\ref{fig:illconditioning_diagram}. Without loss of generality, we set $V_1$ to $1\,\mathrm{p.u.}$ The covariance matrix $\Sigma_V$ is obtained by averaging the outer products of the form below over a large number of samples:
\begin{equation*}
VV^\mathsf{H}=\mathbbl{1}_3\mathbbl{1}_3^\top+
    \begin{bmatrix}
        0 & \Delta v_{12}^* & \Delta v_{13}^* \\
         \Delta v_{12} & \Delta v_{12}^*+\Delta v_{12} & \Delta v_{13}^*+\Delta v_{12} \\
         \Delta v_{13} & \Delta v_{12}^*+\Delta v_{13} & \Delta v_{13}^*+\Delta v_{13}
    \end{bmatrix},
\end{equation*}
where $\Delta v_{1j}=z_{1j}\Delta i_{1j},\forall j\in\{2,3\}$.
The cross-product terms of voltage drops are neglected. The matrix is close to singularity in two cases: (i) if $\Delta v_{12}\approx\Delta v_{13}$ or (ii) if the voltage drops $\Delta v_{12},\Delta v_{13} \ll 1$ are close to the machine precision. The first condition occurs for similarly loaded lines, and the second in light loading conditions. 
\end{example}

Following \cite{WellCondLMMSE2022}, we say that a matrix is \textit{$L$-well-conditioned} if it can be computed without any inverse larger than $L\times L$. A well-conditioned Wiener filter solution can be established by truncating the smallest eigenvalues of the joint covariance matrix. Furthermore, it is well known that discarding the smallest eigenvalues and the corresponding eigenvectors leads to denoising. The eigendecomposition of the joint covariance matrix can be performed to obtain
\begin{equation}
    {\Sigma_{Z}} = \begin{bmatrix} {X_I} \\ {X_V} \end{bmatrix} {S_Z} \begin{bmatrix} {X_I} \\ {X_V} \end{bmatrix}^\mathsf{H},
\end{equation}
with ${X_Z}\coloneqq({X_I},{X_V})$ and eigenvalues ordered from largest to smallest. Now let us partition the eigenvector matrices into an $n\times L$ and an $n\times M$ matrix such that
\begin{equation*}
    {X_I} = \begin{bmatrix} {X}_{{I},L} & {X}_{{I},M} \end{bmatrix},\,\,
    {X_V} = \begin{bmatrix} {X}_{{V},L} & {X}_{{V},M} \end{bmatrix},
\end{equation*}
with $L+M = 2n$ and $L\leq n$. Similarly, $S_Z = \blkdiag(S_{Z,L},S_{Z,M})$ is partitioned into two square diagonal matrices of sizes $L\times L$ and $M\times M$. Furthermore, we define the Karhunen-Lo\`eve transform of ${Z}$ by ${K_Z} \coloneqq {X_Z}^\mathsf{H}{Z}$, from where ${Z} = {X_Z}{K_Z}$. In terms of subvectors we have ${{I}}= {X_I}{K_Z}$ and ${{V}}={X_V}{K_Z}$. Now let ${K}_{{Z},L}$ be the top $L$ submatrix of ${K_Z}$ so that ${{I}}\approx{X}_{{I},L}{K}_{{Z},L}$ and ${{V}}\approx {X}_{{V},L}{K}_{{Z},L}$. A least squares approximation of the transform is given by ${K}_{{Z},L}\approx ({X}_{{V},L}^\mathsf{H}{X}_{{V},L})^{-1}{X}_{{V},L}^\mathsf{H}{{V}}$. Using this estimate, we can obtain a simple approximate filter:
\begin{equation}\label{eq:WCWF}
    {Y_\mathrm{WCWF}} = {X}_{{I},L} ({X}_{{V},L}^\mathsf{H}{X}_{{V},L})^{-1}{X}_{{V},L}^\mathsf{H},
\end{equation}
which is well-conditioned, i.e., the matrix inverses are $L\times L$.
\begin{lemma}
Let $\rho_L\coloneqq \trace{(S_Z)-\trace{(S_{Z,L})}}$ define the truncation power loss. The derived $Y_\mathrm{WCWF}$ filter converges to the Wiener filter $Y_\mathrm{W}$ as $\rho_L\rightarrow0$.
\end{lemma}
The proof of the lemma above is available in \cite{WellCondLMMSE2022}. The preceding lemma demonstrates that the obtained filter corresponds to the Wiener filter in the limit, despite being well-conditioned while the Wiener filter may not be.

\subsection{Embedding the Laplacian Structure via Postfiltering}
As discussed in Sec.~\ref{sec:2}, the admittance matrix is symmetric, and additionally, if the shunt admittances are neglected or nonexistent, $Y$ has zero row-sums. These properties are not guaranteed to hold for the Wiener filter estimate \eqref{eq:WF} or its well-conditioned counterpart \eqref{eq:WCWF}. In this section, we derive a simple postfiltering procedure that can be used to enforce the Laplacian matrix structure. 

Since $Y$ has a known structure, some entries are redundant in the sense that they can be deduced from this structure. Firstly, the admittance matrix is symmetric, thus requiring solely $n_d\coloneqq\frac{1}{2}n(n+1)$ elements to be stored in a vector $\vech{(Y)}$ such that $\vect{(Y)} = D \vech{(Y)}$, where $D\in\{0,1\}^{n^2\times n_d}$ is a full rank matrix called the duplication matrix and $\vect{(Y)}$ is the column vector stacking the columns of $Y$. Furthermore, the diagonal elements are redundant as they can be expressed as a negative sum of the off-diagonal elements in each row. That is, $\vech{(Y)} = R \vechrs{(Y)}$, where $R\in\{-1,0,1\}^{n_d\times n_r}$, with $n_r\coloneqq\frac{1}{2}n(n-1)$, is also full rank and $\vechrs{(Y)}$ collects the off-diagonal elements.

\begin{problem}
Let us assume that $\bar{Y}$ is an admittance matrix estimate obtained via \eqref{eq:opt_reformulated} or \eqref{eq:WCWF}. We consider a problem of determining $\hat{{Y}}$ that is symmetric, has zero row-sums, and is closest to $\bar{Y}$ in the Frobenius norm sense, i.e.,
\begin{align}\label{eq:pfilter}
    \hat{{Y}} =&\argmin_{{Y}\in\C^{n\times n}} \|{\bar{Y}} - {Y} \|_\mathrm{F}^2 \\
    & \quad \mathrm{s.t.} \quad {Y} = {Y}^\top, {Y}_{ii} = -\sum_{j\neq i} {Y}_{ij}, \forall i. \nonumber
\end{align}
\end{problem}

Interestingly, a closed-form solution based on pseudoinverses of $D$ and  $R$ can be obtained.  
\begin{proposition}\label{prop:postfiler}
The offdiagonal entries of $\hat{Y}$, the solution to \eqref{eq:pfilter}, are given by $\vechrs{(\hat{{Y}})} = R^\dagger {D}^\dagger \vect{(\bar{Y})}$.
\end{proposition}
\begin{proof}
By applying the $\vect{(\cdot)}$ operator on the objective function we obtain $\|\bar{Y} - Y \|_\mathrm{F}^2 = \|\vect{(\bar{Y})} - \vect{(Y)} \|_2^2$. The constraints can be included by considering that $\vect{(Y)} = DR\vechrs{(Y)}$. Upon substituting the previous equality in the objective function, the least squares solution is given by $\vechrs{(\hat{Y})} = (DR)^\dagger \vect{(\bar{Y})} = R^\dagger D^\dagger \vect{(\bar{Y})}$. The last equality holds as $R$ and $D$ are full rank.
\end{proof}

Therefore, applying $R^\dagger D^\dagger$ as a postfilter to an obtained $Y$ estimate enforces the Laplacian structure. We note that the pseudoinverses can be constructed efficiently as their structure is generic \cite{NeudeckerElimMatrix1980}, and depends only on the number of buses in the network. To further motivate and justify use of the postfilter, let us consider the following problem.

\begin{problem}
Using the available measurements $\mathbbl{\tilde{I}}$ and $\mathbbl{\tilde{V}}$ a constrained least squares network identification problem respecting the Laplacian structure of $Y$ is formulated as
\begin{align}\label{eq:ls_recast}
    \hat{{Y}} =&\argmin_{{Y}\in\C^{n\times n}} \|\mathbbl{\tilde{I}} - Y\mathbbl{\tilde{V}} \|_\mathrm{F}^2 \\
    & \quad \mathrm{s.t.} \,\, {Y} = {Y}^\top, {Y}_{ii} = - \sum_{j\neq i} {Y}_{ij}, \forall i. \nonumber
\end{align}
\end{problem}

Solution to the constrained least squares problem is given in the proposition below. The proof is omitted as it resembles the proof of Proposition~\ref{prop:postfiler}, with the additional identity $\vect{(Y\mathbbl{\tilde{V}})}=(\mathbbl{\tilde{V}}^\top \otimes \mathcal{I}_n) \vect{(Y)}$ required.

\begin{proposition}
The off-diagonal entries of the solution to \eqref{eq:ls_recast} are given by
\begin{align}
    \vechrs{(\hat{{Y}})} &= {R}^\dagger {D}^\dagger\,(\mathbbl{\tilde{V}}^\top \otimes \mathcal{I}_n)^\dagger \vect{(\mathbbl{\tilde{I}})}\\
    &= {R}^\dagger {D}^\dagger \vect{(\hat{{Y}}_\mathrm{LS})},
\end{align}
where $\hat{{Y}}_{LS}\coloneqq(\mathbbl{\tilde{V}}^\top \otimes \mathcal{I}_n)^\dagger \vect{(\mathbbl{\tilde{I}})}$ denotes the unconstrained least squares solution.
\end{proposition}

The result above shows that the optimal solution to the constrained least squares problem \eqref{eq:ls_recast} can be obtained by applying postfiltering \eqref{eq:pfilter} to the unconstrained least squares solution. The equivalence between the unconstrained least squares solution and the Wiener filter solution \eqref{eq:WF} when the same measurements $\mathbbl{\tilde{I}}$ and $\mathbbl{\tilde{V}}$ are used to compute the sample covariances motivates the use of the postfiltering for ensuring the Laplacian structure in our proposed solution \eqref{eq:WCWF}. 
Note that the least squares problem in \eqref{eq:ls_recast} was previously considered in \cite{IPF2022}. However, the connection between the unconstrained and constrained least squares problems via the postfilter in \eqref{eq:pfilter} was not recognized. 

\begin{figure}[!t]
    \centering
    \includegraphics[scale=0.5]{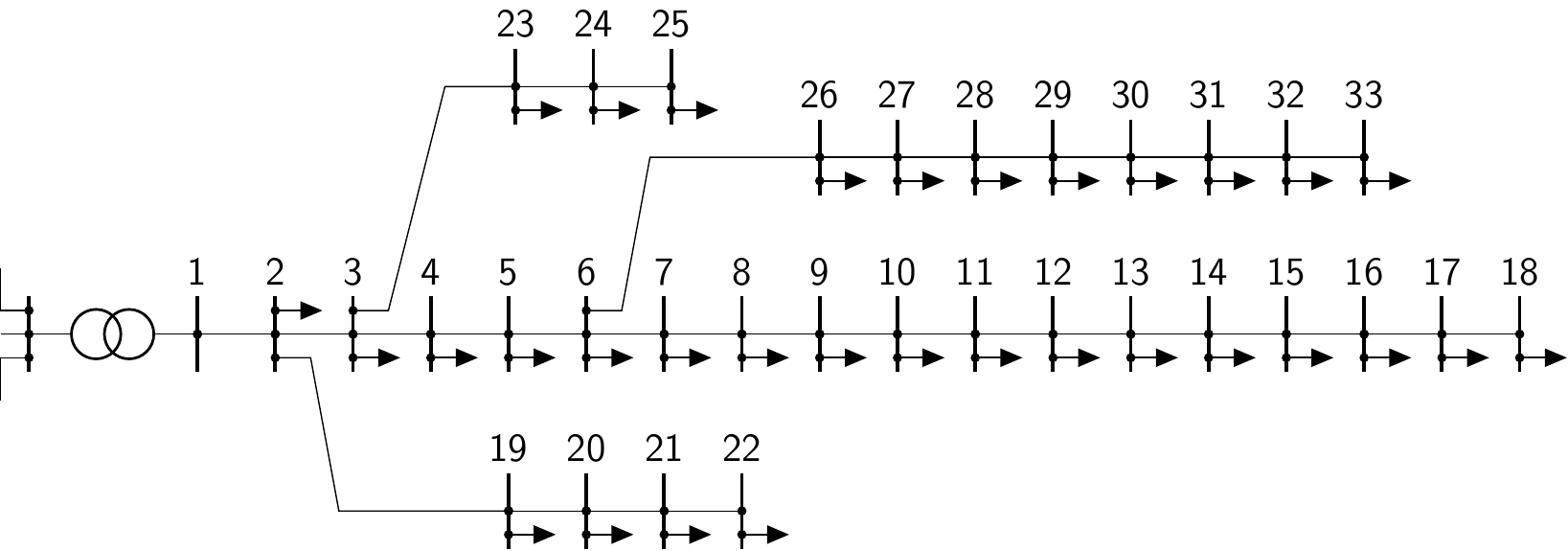}
    \caption{Single line diagram of the IEEE 33-bus test system.}
    \label{fig:33busIEEE}
    \vspace{-0.35cm}
\end{figure}
\begin{figure}[!b]
    \centering
    \includegraphics[scale=1.125]{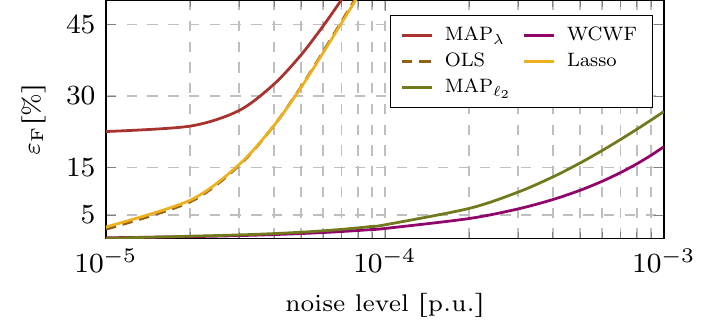}
    \caption{Comparison of the relative error of existing and proposed methods for various noise levels.}
    \label{fig:comp}
\end{figure}

\section{Results}
The proposed identification methods are evaluated on the IEEE 33-bus network presented in Fig.~\ref{fig:33busIEEE}. We assume that a $\mu$PMU device is placed on each node in the network, measuring both voltage and current phasors. The procedure to generate the estimation data follows \cite{JS_TSG2022}, where synthetic load profiles are created using the GENETX generator, and the power flow procedure is run using the PandaPower library. To realistically represent $\mu$PMU measurements, the voltage and current phasors are corrupted with $0.01\%$ standard deviation Gaussian noise. Measurements collected at 50Hz frequency are averaged over a minute and 7 days of thus constructed data ($10080$ tuples of voltage and current phasors) are used in the considered estimation procedures.  

\subsection{Estimation Performance under Varying Noise Levels}
In this section, we compare the estimation accuracy of the proposed methods: (i) identification procedure in \eqref{eq:opt_reformulated}, labeled by $\mathrm{MAP}_\lambda$ henceforth, and (ii) the well-conditioned Wiener filter \eqref{eq:WCWF}, denoted as $\mathrm{WCWF}$ hereafter, to the state-of-the-art approaches from the literature under varying noise levels. More precisely, we use Ordinary Least Squares (OLS) \cite{IPF2022}, Lasso \cite{Ardakanian2019}, and MAP with $\ell_2$ regularization \cite{JS_TSG2022} for benchmarking. The metric used to evaluate the accuracy of an estimation procedure is the relative Frobenius norm $\varepsilon_\mathrm{F} = \|\hat{Y}-Y\|_\mathrm{F}/\|Y\|_\mathrm{F}$, where $Y$ is the \textit{true} admittance matrix and $\hat{Y}$ denotes the estimate. The results in Fig.~\ref{fig:comp} show significant estimation bias when performing \eqref{eq:opt_reformulated} due to the violation of stationarity and constant $x/r$ ratio assumptions. Furthermore, the non-errors in variables models, OLS and Lasso, demonstrate high sensitivity to the measurement noise. Finally, the $\ell_2$ regularized MAP from \cite{JS_TSG2022} and the proposed well-conditioned Wiener filter \eqref{eq:WCWF} demonstrate satisfactory performance over a large range of noise levels. 

\begin{table}[!t]
\renewcommand{\arraystretch}{1.5}
\caption{Accuracy and computation time of the considered estimation methods on three distribution grid test cases.}
\noindent
\centering
    \begin{minipage}{\linewidth} 
    \begin{center}
        \begin{tabular}{ c || c  c | c  c | c  c | } \label{tab:cmptAcc}
            \multirow{2}{*}{\textbf{Network}} & \multicolumn{2}{c}{Lasso} \vline height 2ex & \multicolumn{2}{c}{$\mathrm{MAP}_{\ell_2}$} \vline height 2ex &  \multicolumn{2}{c}{WCWF} \vline \\ 
            \cline{2-7}
            & $\varepsilon_\mathrm{F} [\%]$ & $\tau [s]$ & $\varepsilon_\mathrm{F} [\%]$ & $\tau [s]$ & $\varepsilon_\mathrm{F} [\%]$ & $\tau [s]$ \\
            \cline{1-7}
            \textbf{CIGRE10} & $51.8$ & $0.5$ & $10.8$ & $33$ & $1.9$ & $0.002$ \\
            \textbf{IEEE33}   & $59.5$ & $14.3$ & $2.9$ & $780.5$ & $2.1$ & $0.06$ \\
            \textbf{IEEE123}  & $60.1$ & $125.3$ & $6.36$ & $3527$ & $4.2$ & $0.71$ \\
        \end{tabular}
        \end{center}
    \end{minipage}
    \vspace{-0.35cm}
\end{table}
\begin{table}[!b]
\renewcommand{\arraystretch}{1.5}
\caption{The condition number of matrices subjected to inversion in the least squares and Wiener filter methods.}
\noindent
\centering
    \begin{minipage}{\linewidth} 
    \begin{center}
        \begin{tabular}{ c || c |  c | c } \label{tab:condNR}
            \textbf{Network} & $\kappa(\Sigma_V)$  & $\kappa({\Sigma_I})$ &  $\kappa({X}_{{V},L}^\mathsf{H}{X}_{{V},L})$ \\ 
            \cline{1-4}
            \textbf{CIGRE10}    & $2\times 10^{12}$ & $3\times 10^{10}$ & $1\times 10^{7}$  \\
            \textbf{IEEE33}   & $7\times 10^{12}$ & $4\times 10^{11}$ & $1\times 10^{7}$  \\
            \textbf{IEEE123}  & $1\times 10^{13}$ & $3\times 10^{12}$ & $7\times 10^{6}$\\
        \end{tabular}
        \end{center}
    \end{minipage}
\end{table}
\subsection{Computational Efficiency and Conditioning Analysis}
To analyze the computational efficiency, we perform parameter estimation on three benchmark distribution grids: the 10-bus CIGRE MV feeder, the previously considered IEEE 33-bus system, and the three-phase part of the IEEE 123-bus system consisting of 56 buses. Table~\ref{tab:cmptAcc} summarizes the estimation results in the form of accuracy $\varepsilon_\mathrm{F}$ and computation time $\tau$. Only Lasso, MAP with $\ell_2$ regularization, and the well-conditioned Wiener filter are considered for brevity. The nominal Gaussian noise of $0.01\%$ is used. The table shows that WCWF outperforms the other methods in terms of accuracy and computation time across all tested scenarios.

Table~\ref{tab:condNR} presents the condition numbers of three matrices, namely $\Sigma_V$, $\Sigma_I$, and ${X}_{{V},L}^\mathsf{H}{X}_{{V},L}$, which are subjected to inversion in different identification methods. Specifically, $\Sigma_V$ is inverted in both the OLS approach and the Wiener filter method in \eqref{eq:WF}, $\Sigma_I$ is commonly inverted in the impedance matrix estimation process \cite{MoffatUnsupervised2020}, and ${X}_{{V},L}^\mathsf{H}{X}_{{V},L}$ with $L=n$ is inverted in the proposed well-conditioned Wiener filter approach \eqref{eq:WCWF}. The results indicate that the proposed approach offers significant improvement in conditioning compared to the other methods across all three test cases.

\subsection{MAP Estimation under Stationary Current Injections}
The MAP estimation of $\lambda$ has thus far demonstrated unsatisfactory performance -- see Fig.~\ref{fig:comp}. However, in the previous simulation setting, neither the stationarity nor the constant $x/r$ ratio assumptions were valid. We now enforce the two assumptions and analyze how the performance degrades when deviations from these assumptions are imposed. To this end, a modified version of the IEEE 33-bus system is created by setting $b_{ij}=-r_{ij},\forall (i,j)\in\mathcal{E}$, thus achieving a constant conductance-susceptance ratio of one throughout the network. We introduce operator $\mathrm{ndiag}(\cdot): \C^{n\times n}\rightarrow\C^{n\times n}$, which converts the diagonal entries of a matrix to zeros and keeps the off-diagonal elements. To measure the deviation of a matrix $A\in\C^{n\times n}$ from being unitarily diagonalizable by $\mathbbl{W}$ we define the relative distance $\mathrm{dist}_\mathbbl{W}(A)=\|\mathrm{ndiag}(\mathbbl{W}^\mathsf{H}A\mathbbl{W})\|_\mathrm{F}/\|\mathbbl{W}^\mathsf{H}A\mathbbl{W}\|_\mathrm{F}$.

The results in Fig.~\ref{fig:cov_distance} show the increase in estimation error with the increase in $\mathrm{dist}_\mathbbl{W}(\Sigma_I)$ which quantifies the violation of the stationarity property. The estimator is demonstrated to be unbiased when applied to a constant $x/r$ ratio network. However, a significant bias of approximately $10\%$ error is present when estimating the original IEEE 33-bus network.  A constant estimation error is shown for the original network corresponding to $\Sigma_I=\mathcal{I}_n$ since $\mathbbl{W}$ is not well-defined for non-normal admittance matrices. 

\begin{figure}[!t]
    \centering
    \includegraphics[scale=1.125]{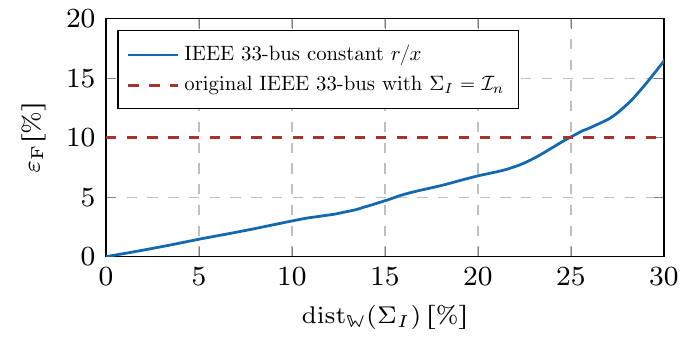}
    \caption{Dependency of the estimation error to the deviation of the current injection covariance matrix to a matrix diagonalizable by $\mathbbl{W}$.}
    \label{fig:cov_distance}
    \vspace{-0.35cm}
\end{figure}

\section{Conclusion}
This paper shows how a maximum a posteriori admittance matrix estimation can be simplified when the current injections are stationary. Nevertheless, the approach only performed well for estimating networks with a constant conductance-susceptance ratio. We have subsequently derived a more general and practical admittance matrix estimation approach based on linear minimum mean square error estimation. Our results demonstrate that the proposed approach is more accurate and computationally efficient than the state-of-the-art when applied to standard test networks.
\bibliographystyle{IEEEtran}
\bibliography{bib.bib}

\end{document}